\documentclass{article}
\usepackage{amsthm}

\usepackage{amssymb}
\usepackage{tikz}
\usetikzlibrary{decorations.pathmorphing}
\usepackage{wasysym}
\usepackage[english]{babel}
\usepackage[a4paper, total={7.5in, 10in}]{geometry}
\newtheorem{theorem}{Theorem}[section]
\newtheorem{lemma}[theorem]{Lemma}

\theoremstyle{definition}
\newtheorem{definition}{Definition}[section]
\theoremstyle{remark}
\newtheorem{remark}{Remark}

\begin{document}

\title{A remark on the energy conditions for Hawking's area theorem}
\author{Martin Lesourd\footnote{martin.lesourd@linacre.ox.ac.uk}, \\ Linacre College, \\ Oxford University }
\maketitle

\abstract{Hawking's area theorem is a fundamental result in black hole theory that is universally associated with the null energy condition. That this condition can be weakened is illustrated by the formulation of a strengthened version of the theorem based on an energy condition that allows for violations of the null energy condition. With the semi-classical context in mind, some brief remarks pertaining to the suitability of the area theorem and its energy condition are made.  }

\section{Introduction} 
Many classic theorems of relativity are obtained by positing a number of local conditions on the geometry of spacetime. These geometric conditions are inequalities imposed, by fiat, on certain contractions of the Einstein or Ricci tensor. With the use of Einstein's equations, these geometric conditions become energy conditions that, supposedly, represent certain energetic characteristics of matter residing in spacetime. It is now understood, however, that these local energy conditions are violated by a number of classical matter models and, moreover, that violations are ubiquitous in the context of quantum field theory in both flat and curved spacetime.\footnote{See \cite{C 14} for a foundational perspective on the status of energy conditions.} To put it another way, the classic theorems aforementioned rely on assumptions not satisfied in contexts considered physically relevant. In view of this, we might wish to ask, for certain theorems of interest, whether they can be formulated with weaker energy conditions. The purpose of this paper is two-fold. To show that Hawking's area theorem can be strengthened as such, and, with the semi-classical context in mind, to interpret the result presented. The definitions used here follow Wald \cite{W 84}.   \\ \\ 
It is instructive, before delving into the area theorem, to consider the well known singularity theorems.\footnote{See \cite{W 84} for an introduction.} There is, it seems, a common template to these theorems. Their assumptions usually include an energy condition, a restriction on the causal properties of the spacetime, and an initial or boundary condition, and their conclusions almost always involve no more than the failure of non-spacelike geodesic completeness. Raychaudhuri's equation describing geodesic congruences is very useful in many proofs of such theorems. In the four dimensional null irrotational case it reads \begin{equation} \frac{d\theta}{d\lambda} = -\frac{1}{2} \theta^2 - \sigma^2 - R_{ab}k^a k^b \end{equation} with \(\lambda\) the affine parameter, \(\sigma\) the shear, \(\theta\) the expansion and \(R_{ab}k^a k^b\) the Ricci tensor twice contracted with a null vector \(k^a\) tangent to the geodesic. Assuming the null convergence condition, \(R_{ab}k^a k^b \geq 0\), or, with Einstein's equations in four dimensions, the null energy condition (NEC), \(T_{ab}k^a k^b\geq 0\), it follows that if the expansion \(\theta\) satisfies \(\theta(\lambda_0) < 0\), then \(\theta \to -\infty \) within finite affine parameter \(\lambda \in (\lambda_0,\infty)\). This behavior is sometimes referred to as geodesic focusing. The onset of geodesic focusing signals the failure of certain geodesics to satisfy certain properties which, in other circumstances, are associated with them. In the null case, a null geodesic focusing to the future of a point signals the geodesic's failure to remain on the boundary of the causal future of that point, and in the timelike case, focusing signals the geodesic's failure to maximize proper time. Many proofs of singularity theorems work by setting up a contradiction under the assumption that all null or timelike geodesics are complete. The rough template is runs as follows. Assume that all null (timelike) geodesics are complete and deduce, under the energy and boundary or initial conditions, the onset of geodesic focusing. Then, combining causal restriction and initial or boundary condition, show that the focusing produced leads to a contradiction. Deduce, therefore, that not all null (timelike) geodesics can be complete.  \\ \\ 
Tipler \cite{T 78} was among the first to show that singularity theorems may be strengthened by way of weaker energy conditions. He defined these weaker energy conditions as non-local restrictions on the integral, along certain types of null or timelike geodesics, of various contractions of the Ricci or Einstein tensor. He showed that these conditions were sufficient to cause focusing and indeed this made it possible to strengthen certain singularity theorems without major amendments to the original arguments. His observation was developed over many years and there arose a number of weaker energy conditions falling under the umbrella term of average energy conditions. One example that continues to generate interest is the average null energy condition (ANEC), which, roughly speaking, is the requirement that \[\int_\gamma R_{ab} k^a k^b d\lambda \geq 0 \] for some suitable class of null geodesics \(\{ \gamma \}\) with \(k^a \) a null vector tangent to the geodesic.\footnote{More care in definitions is taken when we come to formulate the result. }  \\ \\ Though mathematically weaker, the physical interpretation of these average energy conditions is still murky at best. Over the course of a long list of studies, it has been found that many of the average energy conditions allowing for theorem-strengthening are violated by classical matter models, and, less straightforwardly, in the context of quantum field theory in curved spacetime (QFCTS).\footnote{See \cite{C 14, V, F 05} for an expression of this fact and links to some of the relevant references.} To put it another way, many classic theorems of relativity do not, at present, cover a whole host of classical and quantum matter models of physical interest. Efforts to better this situation are, naturally, still ongoing.   \\ \\
In the QFTCS context, there is growing evidence that the extent of the violation of certain energy conditions is, in some sense, restricted. There is a whole body of work dedicated to making this more precise. The idea is to produce certain kinds of inequalities that represent the spatiotemporal constraints that (contractions of) renormalized stress energy tensors in various contexts of QFTCS obey.\footnote{See \cite{F 05} for an introduction and links to various relevant references.} These inequalities are known as quantum energy inequalities (QEI), and, in the best of cases, they have been used either to prove the ANEC in certain circumstances or to constrain the properties of certain spacetime scenarios otherwise associated with violations of certain more standard energy conditions. Examples include: Ford and Roman's study of the properties of traversable wormholes \cite{FR 96}, Fewster, Olum, Pfenning QEI-proof of the ANEC under certain circumstances \cite{Few}, and Kontou and Olum QEI-proof under different circumstances \cite{Kontou}.  \\ \\ 
Yet despite providing valuable insights, QEI have not yet permitted the extension of certain classic relativity theorems to the semi-classical context. Galloway and Fewster \cite{FG 11} recently proposed a study aiming to provide a step in this direction. They formulated versions of Hawking's cosmological and Penrose's collapse singularity theorems based on energy conditions allowing for, respectively, violations of the strong and null energy conditions. Though inspired by QEI methods, the conditions that stand in as energy conditions in their singularity theorems are not, strictly speaking, QEI. On this point, it is worth noting that the relationship between QEI and energy conditions remains not particularly well understood, with only a few definitive results available. One example is \cite{FR 03} where Fewster and Roman show that the existence of QEI is not necessary for the ANEC to be satisfied. \\ \\ Galloway and Fewster's arguments rely on lemmas that establish sufficient conditions for focusing, which they use to strengthen the singularity theorems of, respectively, Penrose and Hawking. Moreover, they do so without having to make any essential amendments to the original spirit of the proof of these theorems, and they also consider a specific matter model to illustrate how the NEC may be violated whilst satisfying their weakened energy conditions. The specific lemma that we shall use is the following. \begin{lemma} Consider the initial value problem for \(z(t)\) \[ \dot{z} = \frac{z^2}{s} + r \] where \(r(t)\) is continuous on \([0,\infty)\), \(z(0)=z_0\) and \(s>0\) is constant. If there exists \(c \geq 0\) such that \[z_0 -\frac{c}{2} + \lim_{T\to \infty} \inf \int_0^T e^{-2ct/s}r(t) dt >0\] then the initial value problem has no solution on \([0,\infty)\), where `no solution' means \(z(t) \to \infty\) as \(t\to t^-_* < \infty\). \end{lemma}  

\section{A stronger area theorem} 
Hawking's area theorem \cite{H} is often considered to be one of the most important results in black hole theory. It describes a fundamental property of dynamical black holes, it underlies black hole thermodynamics and it bounds the amount of radiation that can be emitted upon black hole collisions. Consider Ashketar and Krishnan's remarks in their recent very well cited review on generalized black holes \cite{AK}. \begin{quotation} For fully dynamical black holes, apart from the `topological censorship' results which restrict the horizon topology [...], there has essentially been only one major result in exact general relativity. This is the celebrated area theorem proved by Hawking in the early seventies [...]: If matter satisfies the null energy condition, the area of the black hole event horizon can never decrease. This theorem has been extremely influential because of its similarity with the second law of thermodynamics. \end{quotation} The precise statement of the area theorem depends on a number of definitions for which there are a number of possible choices. The definitions and the proof used in theorem below are identical to those used by Wald \cite{W 84}. This choice is made for reasons of expediency. Note, however, that the following arguments apply to other available formulations of the area theorem. Consider now the main result of this article. \begin{definition} A spacetime \((M,g)\) satisfies the \textit{damped averaged null energy condition} (dANEC) if along each future complete affinely parametrized null geodesic \(\gamma:[0,\infty)\to M\), there exists a non-negative constant \(c \geq 0\) such that \[ \lim_{T\to \infty} \inf \int_0^T e^{-ct} Ric(\gamma', \gamma') dt - \frac{c}{2} >0 \] where \(\gamma'\) is a tangent vector for \(\gamma\). 
\end{definition}
\begin{theorem} 
Let \((M,g_{ab})\) be a four-dimensional strongly asymptotically predictable spacetime that satisfies the dANEC. Let \(\Sigma_1\) and \(\Sigma_2\) be spacelike Cauchy surfaces for the globally hyperbolic region \(\tilde{V}\) with \(\Sigma_2 \subset I^+(\Sigma_1)\) and let \(\mathcal{H}_{1(2)}=H \cap \Sigma_{1(2)}\), where \(H\) denotes the event horizon, i.e., the boundary of the black hole region of \((M,g_{ab})\). Then the area of \(\mathcal{H}_2\) is greater or equal to the area of \(\mathcal{H}_1 \) \end{theorem} 

\begin{proof} The argument is a straightforward application of Galloway and Fewster's lemma 1.1 with the original arguments by Hawking, which are as in \cite{W 84} in this formulation. We start by showing that the expansion \(\theta\) of the null generators of \(H\) is everywhere non-negative. \\ \indent Suppose \(\theta <0\) at \(p\in H\). Let \(\Sigma\) be a spacelike Cauchy surface for \(\tilde{V}\) passing through \(p\) and consider the two-surface \(\mathcal{H}=H\cap \Sigma\). Since \(\theta<0\) at \(p\), we can deform \(\mathcal{H}\) outward in a neighborhood of \(p\) to obtain a surface \(\mathcal{H}'\) on \(\Sigma\) which enters \(J^-(\mathcal{J}^+)\) and has \(\theta<0\) everywhere in \(J^-(\mathcal{J}^+)\). However, by the same argument as in proposition 12.2.2 of \cite{W 84}, this leads to a contradiction as follows. Let \(K\subset \Sigma\) be the closed region lying between \(\mathcal{H}\) and \(\mathcal{H}'\), and let \(q\in \mathcal{J}^+\) and with \(q\in \dot{J}^+(K)\). Then the null geodesic generator of \(\dot{J}^+(K)\) on which \(q\) lies must meet \(\mathcal{H}'\) orthogonally. \\ \indent However, this is impossible since, by \(\theta<0\) on \(\mathcal{H}'\), the null generator of \(\dot{J}^+(K)\), eg., \(\gamma\), will have a conjugate point before reaching \(q\). This follows from considering Raychaudhuri's equation for irrotational null congruences, equation (1), and identifying it with the equation of lemma 1.1. That is, take \(s=2\), \(z=-\theta(t)\), \(r(t)=Ric(\gamma', \gamma')+ 2\sigma^2\) where \(\theta(t)\) is the null expansion relevant to the null geodesic \(\gamma\). We now have an initial value problem as in lemma 1.1 with \(z(0)=-\theta(0)>0\). By the dANEC and the non-negativity of \(\sigma^2\), it follows that the conditions of lemma 1.1 are satisfied, and so there occurs a focal point within finite affine parameter from \(\gamma(0)\). \\ \indent It now follows that \(\theta\geq0\). By standard properties \cite{W 84}, each \(p\in \mathcal{H}_1\) lies on a future inextendible null geodesic, \(\eta\), contained in \(H\). Since \(\Sigma_2\) is a Cauchy surface, \(\eta\) must intersect \(\Sigma_2\) at a point \(q\in \mathcal{H}_2\). Thus, we obtain a natural map from \(\mathcal{H}_1\) into a portion of \(\mathcal{H}_2\). Since \(\theta\geq0\), the area of the portion of \(\mathcal{H}_2\) given by the image of \(\mathcal{H}_1\) under this map must be at least as large as the area of \(\mathcal{H}_1\). In addition, since the map need not be onto, the area of \(\mathcal{H}_2\) may be even larger. Thus, the area of \(\mathcal{H}_2\) cannot be smaller than than of \(\mathcal{H}_1\). \end{proof} 

\begin{remark} We have chosen an energy condition based on the work of Galloway and Fewster because we find it to be both state of the art and neatly amenable to classical field models violating the NEC. In section 6 of their paper \cite{FG 11}, they construct such a model which, by inessential modifications, is straightforwardly applicable to theorem 2.1. We note that other conditions similar in form to the ANEC for semi-complete geodesics could have been used, eg., Roman's condition \cite{R 88} in his version of Penrose's theorem.\end{remark} 

\begin{remark} Recent studies \cite{CDHG} have shown that there are, in both Hawking \cite{H} and Wald's \cite{W 84} formulation of the area theorem, lacunas in the form of unstated or undesirable assumptions of differentiability of the horizon. These deficiencies have been overcome in formulations of much more sophisticated area theorems \cite{CDHG}. Indeed, the authors in \cite{CDHG} mention that the assumptions of the area theorem in \cite{W 84} are satisfied under their conditions. It is worth noting that the strengthening offered here generalizes to these superior area theorems.
\end{remark}
\begin{remark} The exponential damping in the integrand shows that boundary effects are crucial in determining whether area non-decrease obtains. Provided there occurs positive contributions to the null contractions of the Ricci tensor near the horizon, the condition may be satisfied even if negative contributions persist for arbitrarily long segments of the null geodesic away from the horizon. That is: the area may be non-decreasing even if there are large violations to the standard ANEC for semi-complete geodesics.  \end{remark} 

\begin{remark}
A number of other standard theorems about black holes can be strengthened in the way described. The list includes results describing the location of trapped surfaces, marginally trapped surfaces, apparent horizons and so on. Our focus on the area theorem is owed to the obvious tension it generates with the idea of Hawking radiation. 
\end{remark}
\begin{remark}
As a final point, this conclusion seems to dovetail nicely with what emerges from the generalized black holes framework.\footnote{See \cite{J} for an introduction.} In that framework, area non-decrease results are obtained for trapping, isolated and dynamical horizons upon assuming that the NEC holds within a neighborhood of the relevant horizon. We note, then, that the dANEC will be satisfied if there occurs, in the immediate vicinity of the horizon, sufficiently positive contributions to the integrand. With that being said, the relation between global and quasi-local types of horizons is still yet to be fully understood, and so these area non-decrease behaviors ought not to be identified. The parallel is nevertheless worth underlining.
\end{remark}

\section{Discussion}
There is now a certain conventional wisdom regarding the tension between Hawking radiation and the area theorem. Consider, for instance, the remarks made by Visser \cite{V} and, respectively, Bousso and Engelhart \cite{BE 15}. \begin{quotation} The very fact that Hawking evaporation occurs at all violates the area increase theorem...for classical black holes. This implies that the quantum process underlying the Hawking evaporation process must also induce a violation of one or more of the input assumptions used in proving the classical area increase theorem. The only input assumption that seems vulnerable to quantum violation is the assumed applicability of the null energy condition. \end{quotation} \begin{quotation} Hawking's theorem holds in spacetimes obeying the null curvature condition, \(R_{ab}k^a k^b \geq0\) for any null vector \(k^a\). This will be the case if the Einstein equations are obeyed with a stress tensor satisfying the NEC, \(T_{ab}k^a k^b \geq0\). The NEC is satisfied by ordinary classical matter, but it is violated by valid quantum states (e.g., in the Standard Model). In particular, the NEC fails in a neighborhood of a black hole horizon when Hawking radiation is emitted. Indeed, the area of the event horizon of an evaporating black hole decreases, violating the Hawking area law. \end{quotation} 
Galloway and Fewster show that their energy condition leads, under the relevant conditions, to null geodesic incompleteness. Their argument is classical and standard except for the particular form of the energy condition that is used. One may consider whether their results provide evidence in favor of singularities occurring in the semi-classical context, and, likewise, whether the area theorem is relevant to the study of semi-classical black holes. We recall, on this point, that a precise and rigorous understanding of the sense in which the area theorem fails in the semi-classical context is still unavailable. Nevertheless, in view of the almost unchallenged view that black holes do in fact radiate in various semi-classical contexts, there are at least two kinds of possible attitudes: either none of the standard arguments are fully applicable and we must await further progress, or, the standard arguments apply in the sense that it is only the energy condition component of the various theorems that need generalization to the QFTCS context. In the former case, the theorem above provides no particularly new information. As for the latter, there seem to be two possibilities. There may arise conflicting results whereby a model satisfies an energy condition permitting an area theorem (similar, perhaps, to the one above) despite being one for which evaporation is expected. In that case, then, the typical argument in favor of evaporation cannot be sustained. The other possibility is that everything remains harmonious in the sense that models satisfying energy conditions that lead to an area theorem end up being precisely those in which area decrease by evaporation is not expected. Results in either direction would be fruitful.  \\ \\ The view adopted by the author is that neither case seems to be substantiated by precise and rigorous arguments, and, moreover, that current understanding is still some way away from identifying exactly what in the area theorem fails semi-classically. Further understanding of this issue is likely to be provided by a better understanding of a number issues including, for instance, the type of energy conditions that are suitable in the context of QFTCS, the effects of back-reaction, the trans-Planckian problem, and the various links between event horizons and analogous entities occurring in the generalized black holes framework.

\section*{Acknowledgements}
I thank the Ruth and Nevil Mott scholarship and the AHRC for funding this research. I thank Prof. Erik Curiel for his guidance and for writing the paper that provided the initial inspiration for this work \cite{C 14}. I also thank my supervisor, Prof. Harvey Brown, for his support, and last but not least, I thank the reviewers for their insightful suggestions.

\end{document}